\documentclass[letterpaper]{article} 
\usepackage{aaai2026} 
\usepackage{times}  
\usepackage{helvet}  
\usepackage{courier}  
\usepackage[hyphens]{url}  
\usepackage{graphicx} 
\usepackage{natbib}  
\usepackage{caption} 
\usepackage[linesnumbered,ruled,vlined]{algorithm2e}
\usepackage{multirow}

\usepackage{color}
\usepackage{booktabs}
\usepackage{amssymb}
\usepackage{threeparttable}

\usepackage{amsthm}

\usepackage[switch]{lineno}
\usepackage{color}
\usepackage[noend]{algpseudocode}
\usepackage{diagbox}
\usepackage{float}
\usepackage{amsmath}

\usepackage[linesnumbered,ruled,vlined]{algorithm2e}
\usepackage{comment}

\usepackage{newfloat}
\usepackage{listings}

\DeclareCaptionStyle{ruled}{labelfont=normalfont,labelsep=colon,strut=off} 
\floatstyle{ruled}
\newfloat{listing}{tb}{lst}{}
\floatname{listing}{Listing}

\newtheorem{theorem}{Theorem}
\newtheorem{def1}{Definition}

\title{Exact Optimization for Minimum Dominating Sets}
\author{
    Enqiang Zhu\textsuperscript{\rm 1},
    Qiqi Bao \textsuperscript{\rm 1},
    Yu Zhang \textsuperscript{\rm 2},
    Pu Wu \textsuperscript{\rm 3}\thanks{Corresponding authors},
    Chanjuan Liu \textsuperscript{\rm 4}\thanks{Corresponding authors}
}
\affiliations{

\textsuperscript{\rm 1}Institute of Computing Science and Technology, Guangzhou University, Guangzhou, China\\

\textsuperscript{\rm 2}School of Computer Science, Peking University, Beijing, China\\

\textsuperscript{\rm 3}Hithink RoyalFlush Information Network Co., Ltd., Hangzhou, China\\

\textsuperscript{\rm 4}School of Computer Science and Technology, Dalian University of Technology, Dalian, China\\

}

\usepackage{bibentry}

\begin{document}

\maketitle

\begin{abstract}
The Minimum Dominating Set (MDS) problem is a well-established combinatorial optimization problem with numerous real-world applications. Its \text{NP}-hard nature makes it increasingly difficult to obtain exact solutions as the graph size grows. This paper introduces ParDS, an exact algorithm developed to address the MDS problem within the branch-and-bound framework. ParDS features two key innovations: an advanced linear programming technique that yields tighter lower bounds and a set of novel reduction rules that dynamically simplify instances throughout the solving process. Compared to the leading exact algorithms presented at IJCAI 2023 and 2024, ParDS demonstrates theoretically superior lower-bound quality. Experimental results on standard benchmark datasets highlight several significant advantages of ParDS: it achieves fastest solving times in 70\% of graph categories, especially on large, sparse graphs, delivers a speed-up of up to 3,411 times on the fastest individual instance, and successfully solves 16 out of 43 instances that other algorithms were unable to resolve within the 5-hour time limit. 
These findings establish ParDS as a state-of-the-art solution for exactly solving the MDS problem.
\end{abstract}

\section{Introduction}
The Minimum Dominating Set (MDS) problem is a well-established combinatorial optimization challenge aiming to identify the smallest subset of vertices in a graph that can dominate all other vertices. This problem arises naturally in various real-world applications, such as biological network analysis \cite{wuchty2014controllability}, epidemic control \cite{zhao2020minimum}, and influence analysis within social networks \cite{dinh2014approximability}. Despite its practical importance, the MDS problem is categorized as \text{NP}-hard \cite{garey1979computers}. It has been shown that for a graph with $n$ vertices, the MDS problem cannot be approximated to a factor better than $\Omega(\ln{n})$, assuming that P $\neq$ NP holds \cite{gast2015inapproximability}. Consequently, finding an optimal solution for general graphs is computationally impractical, particularly when dealing with large graphs.


Various heuristic strategies have been developed to tackle the MDS problem in large-scale instances \cite{lei2020solving,cai2021two,zhu2024dual}. The most effective heuristic approaches typically employ data reduction techniques to preprocess and simplify the problem, generate an initial solution using greedy methods, and refine the solution through heuristic methods to enhance its quality. While heuristic algorithms can effectively address large instances, they do not guarantee optimal solutions.

In contrast, exact algorithms guarantee optimal solutions, though they often demand exponential computational time. The study of exact algorithms offers several advantages, such as the necessity for precise solutions in certain applications, the limitations of approximation algorithms, and a deeper understanding of \text{NP}-complete problems \cite{fomin2009measure,DBLP:conf/aaai/JiangLM17,zhu2022exact,DBLP:conf/ijcai/Xiao22,DBLP:conf/aaai/FioravantesKKMO24,DBLP:conf/aaai/FioravantesGM25}. The earliest non-trivial exact algorithm for MDS was introduced by Fomin et al. \cite{fomin2004exact}, and since then, various enhancements have been proposed \cite{schiermeyer2008efficiency,van2008design,van2011exact}. Currently, the most efficient exact algorithm for MDS achieves a time complexity of \( O(1.4864^n) \) while maintaining polynomial space requirements \cite{iwata2012faster}. Recently, Jiang et al. \cite{jiang2023exact} introduced a branch-and-bound (BnB) algorithm known as EMOS, which utilizes the independence number of the 2-hop graph as a lower bound and integrates a reduction rule from \cite{alber2004polynomial}. Inza et al. \cite{inza2024exact} developed an implicit enumeration algorithm for the MDS problem. Xiong and Xiao \cite{Xiong2024exact} proposed a BnB algorithm that employs linear programming (LP) relaxations as lower bounds to prune the search space.

Despite significant advancements, accurately identifying MDSs remains a challenge, primarily due to the exponential increase in the search space as graph sizes expand. This paper seeks to develop efficient exact algorithms for the MDS problem using the BnB framework. The BnB approach initially applies reduction rules to simplify the problem instance, followed by branching on vertices in search of a solution. Additionally, a pruning technique is implemented to enhance the algorithm's efficiency when a subproblem's lower bound indicates that it cannot yield a solution better than the current best.

\vspace{0.5mm}
\noindent{\textbf{Our Contribution}} We introduce a novel concept involving a graph with its vertex set partitioned into \(k\) partitions, termed the Partition-Dominating Set (PDS). In this scenario, only the vertices in one partition must be dominated by a set of vertices drawn from each of the other partitions. Our proposed BnB algorithm for the PDS problem incorporates a new lower bound derived from an enhanced LP technique, intended to effectively prune the search space, along with innovative reduction rules aimed at minimizing instance sizes. We demonstrate that our lower bounds outperform those presented in prior studies \cite{jiang2023exact,Xiong2024exact}. Additionally, we conduct experiments to evaluate the efficacy of our proposed algorithm. When compared to the state-of-the-art BnB algorithms, our algorithm demonstrates significantly faster performance on standard datasets. 


\vspace{0.5mm}
\noindent{\textbf{Paper Organization}} The remainder of the paper is organized as follows: Section 2 provides the necessary notations and terminologies. Section 3 delves into a detailed description of the PDS problem. In Section 4, we introduce our innovative lower bound technique, and Section 5 outlines the associated reduction rules. Section 6 presents experimental results, and finally, Section 7 offers concluding remarks.

\section{Notations and Terminologies} \label{sec2}
All graphs discussed in this paper are finite, undirected, and simple. Given a graph $G$, its \emph{vertex set} and \emph{edge set} are denoted by $V(G)$ and $E(G)$, respectively. Two vertices are considered \emph{adjacent} in $G$ if they share an edge, and one is said to be a \emph{neighbor} of the other. We use $N_G(v)$ to represent the set of neighbors of $v$ in $G$, and $N_G[v] = N_G(v) \cup \{v\}$. The \emph{degree} of a vertex $v$ in $G$, denoted by $d_G(v)$, is the cardinality of $N_G(v)$, and a \textit{$k$-vertex} is a vertex of degree $k$. We use $\Delta(G)$ and $\delta(G)$ to denote the maximum degree and the minimum degree of $G$, respectively. Clearly, $d_G(v)<|V(G)|$. For convenience, we use $[k]$ to represent the set $\{1,2,\ldots,k\}$  
and specially define $[0] = \emptyset$. 
Let $S$ be a subset of vertices in $G$. We define the set $N_G(S)$ to be the set of neighbors of each vertex in $S$, excluding the vertices in $S$ itself, and $N_G[S]=N_G(S)\cup S$. When the context is clear, we may use $d(v)$, $N(v)$, $N[v]$, $N(S)$, and $N[S]$ to refer to $d_G(v)$, $N_G(v)$, $N_G[v]$, $N_G(S)$, and $N_G[S]$, respectively. We denote by $G-S$ the subgraph of $G$ by deleting vertices in $S$ and their incident edges, and refer to  $G-(V(G)\setminus S)$ as the \emph{subgraph induced} by $S$, denoted by $G[S]$. In particular, when $S$ contains only one vertex $v$, we replace $G-\{v\}$ with $G-v$.   To \emph{identify} a subset $S$ of vertices in $G$ is to replace those vertices with a single vertex $s$ that is connected to all the vertices that were adjacent in $G$ to any vertex of $S$. The vertex $s$ is referred to as the \emph{identified vertex}, and we denote the graph obtained by identifying  $S$ as $\text{iden}_G(S)$. An \emph{independent set} $I$ in $G$ is a subset of vertices such that no two vertices in $I$ are adjacent in $G$.  A vertex $v\in V(G)$ is considered to be \emph{dominated} by a subset $D\subseteq V(G)$ if $v$ is either in $D$ or adjacent to a vertex in $D$, and \emph{undominated} by $G$ otherwise.  We call $D$ a dominating set (DS) of $G$ if all vertices of $G$ are dominated by $D$, and a \emph{partial dominating set} of $G$ otherwise.  The smallest integer $k$ for which $G$ has a DS of cardinality $k$ is called the \emph{domination number} of $G$, denoted by $\gamma(G)$. An optimal DS of $G$ is a DS with cardinality $\gamma(G)$.  The minimum dominating set (MDS) problem refers to finding an optimal DS in a graph.
We call $G$ a \emph{bipartite graph} if $V(G)$ can be partitioned into two independent sets, denoted as $V_1$ and $V_2$, which are called the \emph{parts} of $G$.

\section{Partition-Dominating Set}
In our branching-based exact algorithm for the MDS problem, each vertex in the graph \( G \) is assigned a status: ``branched'' if it has already been considered during the branching process, and ``unbranched'' otherwise. At each branching step, a vertex is either included in the partial solution or excluded. For each vertex \( v \in V(G) \), let \( T_v \) denote the subtree of the search tree that includes all vertices branched prior to \( v \). Each root-to-leaf path in \( T_v \) represents a sequence of decisions leading to a partial DS, which consists of those vertices added to the solution before reaching \( v \). Given a partial solution \( D \), the remaining vertices in \( V(G) \setminus D \) can be classified into four categories (as noted in \cite{jiang2023exact}): Branched and dominated vertices, branched and undominated vertices, unbranched and dominated vertices, and unbranched and undominated vertices. Throughout this paper, we denote the set of undominated vertices as \( UD \), and the set of unbranched vertices as \( UB \). Our objective is to identify the smallest subset \( D' \subseteq UB \) such that \( D' \) dominates all vertices in \( UD \). The union \(D \cup D' \) then forms a DS of $G$. Let \( D^* \) represent the current best DS of \( G \). If \( |D \cup D'| \geq |D^*| \), the corresponding branch can be safely pruned. Intuitively, as \( |D'| \) increases, the likelihood of triggering this pruning condition also rises. However, determining the exact size of a minimum \( D' \) is computationally challenging. We treat this task as a new problem.

\begin{def1}Partition-dominating set (PDS). \label{pds}
Let $G$ be a graph whose vertex set is partitioned into $k$ parts, say $V_1$, $\ldots, V_k$. For each $i\in [k]$, a $V_i$-PDS  is a subset $P \subseteq (V(G)\setminus V_i)$ such that for each $j\in ([k]\setminus \{i\})$  $P\cap V_j \neq \emptyset$ and every vertex in $V_i$ is dominated by $P$. When $N(x) \not \subseteq V_i$ for each $x\in V_i$, $V_i$-PDS does exist, and the $V_i$-partition-domination number, denoted by $\gamma_{p}(G\rightarrow V_i)$, is the smallest integer $k$ for which $G$ has a $V_i$-PDS of cardinality $k$. Given an $i\in [k]$,  the PDS problem asks to determine  $\gamma_{p}(G\rightarrow V_i)$. 
\end{def1}

The PDS problem has applications in resource allocation. Consider a region represented as a graph \( G \), partitioned into \( k \) administrative districts (\( V_1, \ldots, V_k \)), where vertices represent villages and edges indicate direct connections (such as roads). In the event that a district \( V_i \) experiences a disaster (for instance, a flood), it requires deploying resource teams (such as medical, engineering, and logistics teams) from other districts with two main objectives: (1) to ensure that every village in \( V_i \) is covered; and (2) to guarantee that each unaffected district contributes at least one team. The aim is to determine the minimum number of teams needed to satisfy both conditions, thereby promoting efficient resource allocation and fostering collaboration between districts. This situation is directly linked to computing \( \gamma_p(G \to V_i) \).

\begin{theorem}
The PDS problem is an \text{NP}-hard problem.
\end{theorem}
\begin{proof}
It is enough to consider the decision form of the PDS problem: Given a graph $G$ whose vertex set is partitioned into $k$ parts $V_1,\ldots,V_k$ and two integers $i,n$ where $i\in [k]$ and $n>0$, the question is whether $\gamma_{p}(G\rightarrow V_i) \leq n$. We describe a polynomial reduction from 3-SAT to PDS.  Since 3-SAT is \text{NP}-hard, PDS is also \text{NP}-hard. A 3-SAT instance is a conjunctive normal form $f=C_1\land  C_2\land \ldots \land C_m$, where $C_i, i=1,2,\ldots, m$ is a disjunctive clause of three literals on $n$ variables, denoted by $x_1, x_2,\ldots, x_n$. Here, a literal can be a variable $x$ or its negation $\overline{x}$.  The 3-SAT problem asks to determine whether an arbitrary 3-SAT instance is satisfiable.

Let $f$ be a 3-SAT instance defined as above. We construct, in polynomial time, a graph $G$ whose vertex set $V(G)$ and edge set $E(G)$ are defined as follows:   $V(G) = \cup_{i=1}^{n+1}V_i$ where $V_i = \{x_i,\overline{x_i}\}$ for $i\in [n]$ and $V_{n+1} = \{C_i|i\in [m]\}$;  $E(G) = \{x_iC_j|x_i\in C_j, i\in [n], j\in [m]\}\cup \{\overline{x_i}C_j|\overline{x_i}\in C_j, i\in [n], j\in [m]\}$.  We now show that $\gamma_{p}(G\rightarrow V_{n+1})\leq n$ if and only if $f$ is satisfiable. On the one hand, suppose that  $\gamma_{p}(G\rightarrow V_{n+1})\leq n$ and $P$ is a minimum $V_{n+1}$-PDS of $G$.  Since $P\cap V_i\neq \emptyset$ for $i\in [n]$, we have $|P|\geq n$ which implies  $\gamma_{p}(G\rightarrow V_{n+1})=n$. Then, $|P\cap V_i|=1$ for $i\in [n]$. Define a truth assignment $\sigma$ to the $n$ variables as follows: for each $x\in \{x_1,x_2,\ldots, x_n\}$, $\sigma(x)=1$ if $x\in P$ and  $\sigma(x)=0$ if $\overline{x}\in P$. Clearly, every clause $C$ is dominated by a literal (a vertex) in $P$, implying the literal belongs to $C$, and hence  $f$ is satisfied by $\sigma$.  
On the other hand, suppose that $f$ is satisfied by a truth assignment, say $\sigma$. Construct a set $P$ as follows:  $P=\{x_i|i\in [n], \sigma(x_i)=1\}\cup \{\overline{x_i}|i\in [n], \sigma(x_i) = 0\}$. We claim that $P$ is a  $V_{n+1}$-PDS of $G$. If not, suppose that there exists a clause $C$ that is not dominated by $P$. This indicates that no literal in $C$ has truth 1,  which contradicts the assumption that $f$ is satisfied by $\sigma$.
\end{proof}

Indeed, it suffices to consider the PSD problem on a graph with its vertex set partitioned into two distinct subsets. As previously mentioned, to identify a subset \( D' \) from \( UB \) that dominates all vertices in \( UD \), we construct a bipartite graph, denoted as \( G_{UB, UD} \). In this graph, the two parts are \( UB \) and \( UD \), with an edge \( uv \) connecting \( UB \) to \( UD \) if and only if \( u \in UB \), \( v \in UD \), and \( v \in N_G[u] \), where \( G \) represents the input instance. Thus, \( D' \) is a \( UD \)-PDS for \( G_{UB, UD} \). Note that \( UB \cap UD \) may not be empty. When no vertex is branched, meaning that \( UB = UD = V(G) \), it follows that \( \gamma_{p}(G_{UB, UD}\rightarrow UD) = \gamma(G) \). Consequently, the MDS problem can be viewed as a specific case of the PDS problem. 

\section{A Lower Bound of $\gamma_{p}(G_{UB, UD}\rightarrow UD)$} \label{sec-4}
Consider the following relaxed integer program discussed in \cite{Xiong2024exact}. 
\begin{equation}\label{equ-IL}
\small
\begin{split}
\min \quad & \sum\limits_{v\in UB}x_v  \\
\text{s.t.} \quad & \sum\limits_{v\in N(u)}x_v\geq 1, ~~u\in UD\\
&0\leq x_v\leq 1,~~ v\in UB
\end{split}
\end{equation}

One of our objectives is to find a \( UD \)-PDS of \( G_{UB, UD} \) that is as small as possible, thereby enhancing the likelihood of pruning. It is worth mentioning that the PSD problem remains \text{NP}-hard even when constrained to bipartite graphs. This can be demonstrated through a polynomial reduction of the MDS problem to the PDS problem. In the following section, we will outline a method for determining a lower bound for \( \gamma_{p}(G_{UB, UD}\rightarrow UD) \).


We denote the optimal value of Equation (\ref{equ-IL}) as \( sol \). For a \( UD \)-PDS \( D \) of \( G_{UB,UD} \), we define a variable \( x^{D}_v \) associated with each vertex \( v \in UB\) such that \( x^{D}_v = 1 \) if \( v \in D \) and \( x^{D}_v = 0 \) if \( v \notin D \). 
Since $\{x^D_v\mid v\in UB\}$ satisfies all constraints in Equation (\ref{equ-IL}),  \( |D| = \sum_{v \in UB} x^{D}_v \geq \lceil sol \rceil \), which implies that \( \lceil sol \rceil \) is a lower bound for \( \gamma_{p}(G_{UB, UD} \rightarrow UD) \). We call \( \lceil sol \rceil \) a linear program relaxation lower bound. Our approach utilizes several linear programs to establish a more accurate lower bound than \(\lceil sol \rceil\).

Suppose that \( f_1, f_2, \ldots, f_k \) are \(k\) functions from \( UB \) to \( \mathbb{Z} \), where \(\mathbb{Z}\) denotes the set of integers. We consider \( k \) linear programs \( LP(f_1), LP(f_2), \ldots, LP(f_k) \), defined as follows for \( j = 1, 2, \ldots, k \):

\vspace{-0.3cm}
\begin{equation}\label{equ-fIL}
\small
\begin{split}
\min \quad & \sum\limits_{v\in UB}f_j(v)x_v  \\
\text{s.t.} \quad & \sum\limits_{v\in N(u)}x_v\geq 1, ~~u\in UD\\
\quad & \sum\limits_{v\in UB}x_v= \lceil sol\rceil\\
\quad & \sum\limits_{v\in UB}f_i(v)x_v\geq \lceil sol_i\rceil, ~\text{for} ~ i\in [j-1]\\
&0\leq x_v\leq 1,~~ v\in UB
\end{split}
\end{equation}
Here, \( sol_i \) represents the optimal value of the linear program \( LP(f_i) \). 
Specifically, if no feasible solution exists for \( LP(f_i) \), we denote \( sol_i \) as \( +\infty \).
We will prove that  \( \gamma_{p}(G_{UB, UD} \rightarrow UD) \geq \lceil sol \rceil + 1 \) if there exists an integer \( j \in [k] \) for which \( LP(f_j) \) has no feasible solution, thereby yielding a stronger lower bound.

\begin{theorem}\label{lemma-lower-bound}
For the $k (\geq 1)$ linear programs defined in Equation (\ref{equ-fIL}), if there is an integer $j\in [k]$ such that $LP(f_j)$ has no feasible solution,  then $\gamma_{p}(G_{UB, UD}\rightarrow UD)\geq \lceil sol\rceil + 1$.
\end{theorem}

\begin{proof}
Suppose, to the contrary, that there exists an integer $j \in [k]$ such that $LP(f_j)$ has no feasible solution, while there exists a $UD$-PDS $D$ with $|D| = \lceil \mathit{sol} \rceil$. Consider the vector $\textbf{e}$=($x^{D}_v: v\in UB)$, where $x^{D}_v = 1$ if $v \in D$, and $x^{D}_v = 0$ otherwise. We have 
\[
\sum_{v \in N(u)} x^{D}_v \geq 1 \quad \text{for all } u \in UD, \quad \text{and} \quad \sum_{v \in UB} x^{D}_v = \lceil \mathit{sol} \rceil.
\]

We now prove that $\textbf{e}$ is a feasible solution of all $LP(f_j)$ for $j\in [k]$, leading to a contradiction of our initial assumption. 
It suffices to show that $\sum_{v\in UB}f_i(v)x_v\geq \lceil sol_i\rceil$ for $i\in [j-1], j\in [k]$.
We will proceed using induction on $i$.

For the base case of $i=1$, 
since $[i-1] = [0] = \emptyset$, 
$\textbf{e}$ is a feasible solution of $LP(f_1)$.

Now, for $i > 1$, we assume that for all $i' \in [i-1]$, the condition $\sum_{v \in UB} f_{i'}(v) x^{D}_v \geq \lceil \mathit{sol}_{i'} \rceil$ holds. Consequently, $\textbf{e}$ is a feasible solution for $LP(f_i)$, thus leading to $\sum_{v \in UB} f_i(v) x^{D}_v \geq \lceil \mathit{sol}_i \rceil.$
\end{proof}

\noindent\textbf{Selection of functions $f_1,f_2,\ldots,f_k$.} 
While the functions \( f_1, f_2, \ldots, f_k \) are broadly applicable, the selection of specific functions can significantly impact the efficiency of deriving improved lower bounds. In this paper, we define these functions based on empirical trials.

We determine the value of \( k \) using the formula \( k = \lceil \mathit{sol} \rceil \) and equally partition \( UB \) into \( k \) groups randomly. This process involves constructing \( k - q \) groups of vertices, each with cardinality \( p \), denoted as \( A_1, A_2, \ldots, A_{k-q} \), along with \( q \) groups of vertices, each with cardinality \( p + 1 \), denoted as \( A_{k-q+1}, \ldots, A_k \), where \( |UB| = pk +q \) ($q<k$). The groups are mutually exclusive, ensuring that \( A_i \cap A_j = \emptyset \) for any \( i, j \in [k] \) with \( i \neq j \), and it holds that \( A_1 \cup A_2 \cup \ldots \cup A_k = UB \). Each function \( f_j \) for \( j \in [k] \) is then defined as follows: 
\[
f_j(v) = 
\begin{cases}
1, & \text{if } v \in A_j, \\
0, & \text{otherwise}.
\end{cases}
\]

After selecting the functions \(f_1, f_2, \ldots, f_k\), we can establish a new lower bound by solving a series of linear programs. Algorithm \ref{al-lower-bound} outlines this process in detail.

\vspace{-0.3cm}

\begin{algorithm}[h!]
\scriptsize
\caption{$Multi\_LP(G_{UB,UD}, \lceil \mathit{sol} \rceil)$}\label{al-lower-bound}

\KwIn{A bipartite instance $H = G_{UB,UD}$ and the linear programming relaxation lower bound $\lceil \mathit{sol} \rceil$.}

\KwOut{A new lower bound of $\gamma_p(H \rightarrow UD)$.}

$k \leftarrow  \lceil \mathit{sol} \rceil$;

$A_1, A_2, \ldots, A_k \leftarrow$ A random $k$-partition of $UB$ such that $||A_i|-|A_j|| \leq 1$ for any $i,j\in [k]$ with $i \neq j$;


\For{$j = 1$ \KwTo $k$}{
    Define $f_j(v) = 1$ if $v \in A_j$, and $f_j(v) = 0$ otherwise;
}

\For{$i = 1$ \KwTo $k$}{
    Solve the linear program $LP(f_i)$;
    
    \If{$LP(f_i)$ has no feasible solution}{
        \Return{$\lceil \mathit{sol} \rceil + 1$};
    }
}

\Return{$\lceil \mathit{sol} \rceil$};
\end{algorithm}

\vspace{-0.3cm}

\noindent\textbf{Conditions for applying $Multi\_LP(G_{UB,UD})$.}   
In contrast to the previous method \cite{Xiong2024exact}, $Multi\_LP()$ employs multiple linear programs to strive for an enhanced lower bound. Although this strategy is computationally more intensive, it can yield significantly improved bounds that facilitate branch pruning in scenarios where the original bound proves inadequate. To mitigate computational overhead and effectively reduce overall solving time, we strategically implement our method only when empirical conditions indicate a high likelihood of success in producing a better bound that allows for pruning when the original bound fails. Specifically, let $D^*$ represent the current best UD-PDS of $G_{UB,UD}$ and let $r$ indicate the number of vertices that have already been branched upon and included in a UD-PDS for $G_{UB,UD}$. We invoke $Multi\_LP()$ if and only if  $\lceil sol \rceil + r < |D^*|$, $sol + r > |D^*| - 1.1$, $\frac{n}{k} \leq \left( \frac{m}{n} \right)^2$, and $r \geq 8$.


%


\section{Reduction Rules}

There are night reduction rules that correspond to the minimum  $UD$-PDS of a bipartite instance $H=G_{UB, UD}$, where $d_H(v)\geq 1$ for every $v\in UD$.
Each reduction rule assumes that all previous rules cannot be applied. The new rules, marked by an asterisk, were introduced by us. In these rules, $\text{pds}(H)$ refers to a minimum $UD$-PDS of $H$. The equal sign between the two parts indicates the existence of a minimum UD-PDS. For instance, ``$\text{pds}(H_1) = \{s\} \cup \text{pds}(H_2)$'' implies that $H_2$ has a minimum $UD$-PDS $P$ such that $P \cup \{s\}$ is  a minimum $UD$-PDS of $H_1$.

\textbf{Rule 1} If $u\in UD$ with $d_H(u) = 1$ and $u\in N_H(v)$, then $\text{pds}(H) = \{v\}\cup \text{pds}(H-N_H[v])$.

\textbf{Rule 2} If $u\in UB$ with $d_H(u) = 1$, then $\text{pds}(H) =  \text{pds}(H-\{u\})$.

\textbf{Rule 3} If $u_1, u_2\in UB$ with $N_H(u_1)\subseteq N_H(u_2)$, then $\text{pds}(H) =  \text{pds}(H- \{u_1\})$.

\textbf{Rule 4} If $u_1, u_2\in UD$ with $N_H(u_1)\subseteq N_H(u_2)$, then $\text{pds}(H) =  \text{pds}(H- \{u_2\})$.

\textbf{Rule 5*}  If $u\in UD$ with $N_H(u) = \{u_1, u_2\}$ and $d_H(u_1) = d_H(u_2) = 2$, let $N_H(u_1)=\{u,v_1\}$,  $N_H(u_2)=\{u,v_2\}$,  $H' = \text{iden}_{G- N_H[u]}(\{v_1,v_2\})$ with identified vertex $v$. Suppose that $N_H(\{v_1,v_2\})\setminus \{u_1,u_2\} \neq \emptyset$, and let  $D' = \text{pds}(H')$. Then,  $\text{pds}(H) =  \{u_2\}\cup D'$ if $v_1\in N_H(D')$;  otherwise, $\text{pds}(G) =  (\{u_1\}\cup D')$.

\textbf{Rule 6 \cite{van2011exact}} If $u\in UB$ with $N_H(u) = \{u_1,u_2\}$ and $d_H(u_1) = d_H(u_2) = 2$, let $N_H(u_1)$ = $\{u, v_1\}$, $N_H(u_2)=\{u,v_2\}$,  $H' = \text{iden}_{H-N[u]}(\{v_1,v_2\})$ with identified vertex $v$, and  $D = \text{pds}(H')$.
Then, $\text{pds}(H) =  \{u\} \cup D$ if $v\notin D$; otherwise $\text{mds}(H) =  (D\cup \{v_1,v_2\})\setminus \{v\}$.

\textbf{Rule 7 \cite{fomin2009measure}} If $d_H(u)\leq 2$ for every $u\in UB$, then a solution can be computed in polynomial time using maximum matching.

\textbf{Rule 8*} Let $u\in UB$ have 2-degree neighbors $v_1,\ldots, v_k$, $k\geq 1$. Let $N_H(v_i) = \{u, u_i\}$ for $i\in [k]$. If there is a vertex $u_i$ such that $(N_H(u_i)\setminus \{v_i\})\subseteq \cup_{j \in [k]\setminus \{i\}} N_H(u_j)$,
then $\text{pds}(H) =  \{u\}\cup \text{pds}(H-N_H[u])$.

\textbf{Rule 9*}  Let $u\in UD$ such that $d_H(u)=2$ and $N_H(u) =\{v_1,v_2\}$. If there is a vertex $v\in UB\setminus \{v_1,v_2\}$ such that $N_H(v)\subseteq N_H(v_1)\cup N_H(v_2)$, then $\text{pds}(H) =  \text{pds}(H-v)$. 

\begin{theorem}
\textbf{Rules 1-9} are correct.
\end{theorem}
\begin{proof}
The proof of \textbf{Rules 1-4} can be easily verified. \textbf{Rules 6-7} are provided in previous literature. 

Regarding \textbf{Rule 5}, we denote by $UB'$ and $UD'$ the two parts of $H'$, i.e., $UB'=UB\setminus \{u_1,u_2\}$ and $UD'=(UD\setminus \{v_1,v_2\})\cup \{v\}$. Since $D'$ is a $UD'$-$PDS$ of $H'$, $v_1\in N_H(D')$ or $v_2\in N_H(D')$. Therefore, $\{u_2\}\cup D'$ is a $UD$-$PDS$ of $H$ when $v_1\in N_H(D')$ and $\{u_1\}\cup D'$ is a $UD$-$PDS$ of $H$ when $v_1\notin N_H(D')$. This implies that  $\gamma_{p}(H\rightarrow UD)\leq \gamma_{p}(H'\rightarrow UD')+1$. On the other hand, let $D=\text{pds}(H)$. Suppose that $|D|\leq |D'|$. Clearly, $\{u_1,u_2\}\cap D \neq \emptyset$. If $\{u_1,u_2\}\subseteq D$, then $(D\setminus \{u_1,u_2\})\cup \{w\}$ is a $UD'$-PDS of $H'$, where $w$ is an arbitrary vertex in $N_H(\{v_1,v_2\})\setminus \{u_1,u_2\}$. This implies that $\gamma_{p}(H'\rightarrow UD')\leq |D'|-1$, a contradiction.   If $\{u_1,u_2\}\not \subseteq D$, then only one vertex of $u_1$ and $u_2$ belongs to $D$, say $u_1\in D$ by symmetry. In this case, $D\setminus \{u_1\}$ is a $UD'$-PDS of $H'$, implying that $\gamma_{p}(H'\rightarrow UD')\leq |D'|-1$, and a contradiction.

Regarding \textbf{Rule 8}, it suffices to prove that there is a $\text{pds}(H)$ containing $u$. Let $D$ be a minimum $UD$-$PDS$ of $H$. 
If $u\notin D$, then $u_i\in D$ for every $i\in [k]$.
Then, $D'=(D\setminus \{u_i\})\cup \{u\}$ is a $UD$-$PDS$ of $H$ and $|D'|\leq |D|$.

Regarding \textbf{Rule 9}, we will demonstrate that there is a minimum $UD$-$PDS$ of $H$ such that $v\notin PDS$. Let $D$ be a minimum $UD$-$PDS$ of $H$. Given that $u$ is a 2-vertex, it follows that $\{v_1,v_2\} \cap D \neq \emptyset$, say $v_1\in D$. If $v\in D$, then $D'=(D\setminus \{v\})\cup \{v_2\}$ is a minimum $UD$-$PDS$ of $H$ (Note that $v_2\notin D$; otherwise,  $|D'|< |D|$, a contradiction).
\end{proof}

\noindent\textbf{Time Complexity} Given an bipartite instance  $H=G_{UB,UD}$ with $|UB|=n$ and $|UD|=m$, our branch algorithm recursively applies these reduction rules to $H$, as shown in  Algorithm \ref{al-1}.
\textbf{Rule 1} and \textbf{Rule 2} are simple, which can be implemented in $O(m)$ and $O(n)$, respectively. \textbf{Rule 3} (or \textbf{Rule 4}) enumerates each vertex in $UB$ (or $UD$) and its neighbors, which requires $O(|E(H)|\Delta(H))$ time, as detailed in \cite{Xiong2024exact}.
\textbf{Rule 5} is implemented by examining 2-vertices in $UD$ such that their neighbors are also 2-vertices and conducting the identifying operation, which requires $O(n)$ time. Similarly, \textbf{Rule 6} is also implemented in $O(n)$. \textbf{Rule 7} involves finding a maximum matching, which can be done in $O(\sqrt{n+m} |E(H)|)$ \cite{hopcroft1973n}. Regarding \textbf{Rule 8},  we begin by calculating \( c(w) \) for each \( w \in N_H(\{ u_1, \ldots, u_k \}) \), where \( c(w) \) is the number of vertices in \( \{ u_1, \ldots, u_k \} \)  adjacent to \( w \). Initially, $c(w)=0$. For each $i\in [k]$, we examine the neighborhood \( N_H(u_i) \), and each time a vertex \( w \) is found in \( N_H(u_i) \), we increment \( c(w) \) by 1. After processing all  \( u_i \),  \( c(w) \) is determined for every \( w \in N_H(\{ u_1, \ldots, u_k \}) \). Note that \( N_H(u_i) \setminus \{ v_i \} \subseteq \cup_{j \in [k] \setminus \{ i \}} N_H(u_j) \) if and only if \( c(w) \geq 2 \) for all \( w \in N_H(u_i) \setminus \{ v_i \} \). Therefore, we can identify a vertex \( u_i \) satisfying the condition by examining \( N_H(u_i) \) for each $i\in [k]$. This requires \( O(k \Delta(H)) \) time. Since \( k \leq d_H(u) \), the overall time complexity for examining all vertices $u\in UB$ is \( O(|E(H)| \Delta(H)) \). \textbf{Rule 9} enumerates all 2-vertices in \( UD \). Using a similar analysis as above, the overall time complexity is \( O(m (\Delta(H))^2) \). 

\section{The Algorithm ParDS}

Given an instance \( H = G_{UB, UD} \) and a vertex \( s \in UB \), when branching at \( s \), we consider two cases: one in which \( s \) is included, resulting in the sub-instance \( H - N_H[s] \), and one in which \( s \) is excluded, yielding the sub-instance \( H - s \).  We denote \( D \) as the set of vertices included in a \( UD \)-PDS, and \( D^* \) as the current best \( UD \)-PDS of \( H \). The instance is first reduced to \( G_{UB', UD'} \), where \( UB' = UB - D \) and \( UD' = UD - N_H(D) \). 
Then, if \( |D_v| + \ell b(G_{UB', UD'}) \geq |D^*| \), there is no need for further branching on the vertices in \( H - N_H[D] \),  where \( \ell b(G_{UB', UD'}) \) represents a lower bound of \( \gamma_p(G_{UB', UD'} \rightarrow UD') \).
We introduce an algorithm named \( ParDS \) that resolves the PDS problem by integrating the lower bound with reduction rules. The \( ParDS \) algorithm accepts an instance \( H = G_{UB, UD} \), a lower bound \( \ell b(H) \) for \( \gamma_p(H \rightarrow UD) \), two variables \( r \) and \( q \) as inputs. It produces a minimum \( UD \)-PDS of \( H \) if \( \ell b(H) + r < q \); otherwise, it returns \( UB \). In this context, \( q =|D^*| \) and \( r \) denotes the number of vertices that have been branched and included in a \( UD \)-PDS of \( H \). Note that \( ParDS(H, 0, |UB|) \)  yields the minimum \( UD \)-PDS for \( H \).

The pseudocode of ParDS is described in Algorithm \ref{al-1}. If  $\ell b(H)+r\geq q$, then the algorithm can not find a better $UD-PDS$ of $H$ than the current best $UD-PDS$ $D^*$ and the algorithm returns the set $UB$ to show this case (lines 1-2); otherwise, the algorithm enters into the reduction process (lines 3-11). Let $H'=G_{UB',UD'}$ be the reduced graph by applying reduction \textbf{Rule $k$} to $H$ (line 5). If $k=7$, the algorithms directly return a minimum $UD-PDS$ of $H$ (lines 6-7); if $k\in \{1,5,6,8\}$, then  $\gamma_{p}(H\rightarrow UD)=\gamma_{p}(H'\rightarrow UD')+1$ and we consider ParDS$(H',r+1,q)$ (lines 8-9); if $k\in \{2,3,4,9\}$, then $\gamma_{p}(H\rightarrow UD)=\gamma_{p}(H'\rightarrow UD')$ and we consider ParDS$(H',r,q)$ (lines 10-11).

\vspace{-0.2cm}

\begin{algorithm}[h!]
\scriptsize
\caption{ParDS ($G_{UB,UD}$, $\ell b$, $r$, $q$)}\label{al-1}


\KwIn{A bipartite instance $H=G_{UB,UD}$, $\ell b(H)$,
two variables $r$ and $q$}

\KwOut{A minimum $UD$-$PDS$ of $H$ if $\ell b(H)+r<q$, or $UB$ otherwise}

\If{$\ell b(H)+r\geq q$} { \Return{$UB$;} }

\If{A reduction rule can reduce $H$} {

 $k \leftarrow \min \{i\mid i\in \{1,2,\ldots,9\} ~~\text{and \textbf{Rule $i$} is feasible} \}$;

 $H' \leftarrow$ Reduce $H$ by \textbf{Rule $k$};

\If{$k=7$} {  \Return{A minimum $UD$-$PDS$ of $H$}; } 
\ElseIf{$k\in\{1,5,6,8\}$} {\Return{$ParDS(H',r+1,q)$}}
\ElseIf{$k\in\{2,3,4,9\}$}{\Return{$ParDS(H',r,q)$};}  

}

$s \leftarrow$ A vertex chosen for branching using the LDPB strategy; 

 Compute the optimal values $sol(H - N_H[s])$ and $sol(H - s)$ by solving the corresponding linear programs defined in Equation~\ref{equ-IL};  

$\ell b(H - N_H[s]) \leftarrow \lceil sol(H - N_H[s]) \rceil$, $\ell b(H - s) \leftarrow  \lceil sol(H - s) \rceil$;

\If{$\ell b(H - N_H[s]) + r < q$ \text{and} $sol(H - N_H[s]) + r > q - 1.1$ \text{and} $\frac{n}{k}\leq (\frac{m}{n})^2$ \text{and} $r\geq 8$}{
    $\ell b(H - N_H[s]) \gets Multi\_LP(H - N_H[s],\, \ell b(H - N_H[s]))$;
}
\If{$\ell b(H - s) + r < q$ \text{and} $sol(H - s) + r > q - 1.1$ \text{and} $\frac{n}{k}\leq (\frac{m}{n})^2$ \text{and} $r\geq 8$}{
    $\ell b(H - s) \gets Multi\_LP(H - s,\, \ell b(H - s))$;
}




$sol_{t}\leftarrow \{s\}\cup ParDS(H-N_{H}[s],\ell b(H-N_{H}[s]),r+1,q)$;

$sol_{d}\leftarrow Par$-$PDS(H-s,\ell b(H-s),r,q)$;

$sol \leftarrow$  A solution in $\{sol_{t}, sol_{d}\}$ with a smaller cardinality;

\If{$|sol|+r<q$} {
 $q\leftarrow |sol|+r$;~~ \Return{$sol$};}
 
\Else {\Return{$UB$};}
\end{algorithm}

\vspace{-0.4cm}

Next, the algorithm selects a vertex for branching (line 12). To enhance the efficiency of vertex selection, we introduce a Lexicographic Degree-Priority Branching (LDPB) strategy. This strategy incorporates three key parameters: the maximum degree \(d\) among the vertices in \(UB\), the maximum degree \(max3d\) among vertices in \(N_H(UB \cap L_3(H))\), and the set \(C(H)\) of vertices \(v \in UD\) for which \(|N_H(v) \cap L_d(H)| \geq d_H(v) - 1\), where $L_k(H)$ denote the set of vertices in $H$ with degree $k$. For any vertex \(v \in V(H)\), we define \(\text{ascend}(v)\) and \(\text{descend}(v)\) to represent the ascending and descending order of the degrees of \(v's\) neighbors, respectively. If \(d \geq 4\) and \(C(H) = \emptyset\), a vertex \(s\) is chosen randomly from \(UB \cap L_d(H)\). If \(d \geq 4\) and \(C(H) \neq \emptyset\), vertex \(s\) is selected from \(N_H(s') \cap L_d(H)\), where \(s'\) is a vertex in \(C(H)\) with the minimum \(d_H(s')\), and \(\text{ascend}(s)\) has the minimum lexicographical order. When \(d = 3\) and \(max3d \geq 4\), vertex \(s\) is chosen from \(UB \cap L_3(H)\) such that \(\text{descend}(s)\) has the maximum lexicographical order. Lastly, if \(d = 3\) and \(max3d \leq 3\), \(s\) is selected from \(UB \cap L_3(H)\) with the maximum \(|N_H(s) \cap V_2(H)|\).  

After this, the algorithm branches at vertex $s$ and determines the lower bounds of $H-N_H[s]$ and $H-s$ using the proposed $Multi\_LP$ (lines 13-18). It then recursively computes $sol_{t} = \{s\}\cup ParDS(H-N_H[s],\ell b(H-N_H[s]),r+1,q)$ and $sol_{d} = ParDS(H-s,\ell b(H-s),r,q)$ (lines 19-20). Let $sol$ be the smaller one between $sol_{t}$ and $sol_{d}$ (line 21). Clearly, $|sol| = \gamma_{p}(H\rightarrow UD)$. Since $q=|D^*|$, the algorithm updates $q$ by $|sol|+r$ and returns $sol$ if $|sol|+r<q$ (lines 22-23); and returns $UB$ otherwise (lines 24-25).

It should be noted that for an arbitrary instance \(G\), we construct the bipartite graph \(G_{UB,UD}\) with \(UB = UD = V(G)\) and the edge set \(E = \{uv \mid u \in UB, v \in UD, uv \in E(G)\}\). Then, we can invoke \(\mathrm{ParDS}(G_{UB,UD}, 0, 0, |V(G)|)\) to derive an MDS of $G$.

\section{Experiments} \label{sec-6}
This section evaluates the effectiveness of our ParDS algorithm. We start by introducing four benchmark datasets that are commonly referenced in the literature \cite{cai2021two,jiang2023exact,inza2024exact,Xiong2024exact}, detailed as follows.


\textbf{T1} \footnote{https://github.com/yiyuanwang1988/MDSClassicalBenchmarks} : This dataset comprises 53 graph families that vary in scale and density \cite{romania2010ant}. Each family includes 10 randomly connected graphs of the same order and size, generated using different random seeds (0-9).

\textbf{Network Repository} \footnote{https://networkrepository.com}:  We selected 120 graphs from this diverse collection \cite{rossi2015network}, which spans eight families, to assess our method on real-world networks.

\textbf{BD3/BD6} \footnote{https://doi.org/10.17632/rr5bkj6dw5.4}: The BD3 and BD6 families contain 38 random graphs, as noted in \cite{Xiong2024exact}. These graphs were generated by Inza et al. \cite{inza2024exact}.

\textbf{UDG}\footnote{https://github.com/yiyuanwang1988/MDSClassicalBenchmarks} : This dataset models wireless sensor networks and consists of two families of graphs \cite{potluri2011two}.
Each family includes graphs of varying scales, with 10 instances generated at each scale using random seeds (0-9).

\subsection{Baseline Algorithms}
We compare our algorithms with state-of-the-art exact MDS solvers, including EMOS \cite{jiang2023exact} and four BIB‑based algorithms introduced by \cite{Xiong2024exact}.  
Notably, \cite{inza2024exact} proposed an implicit enumeration algorithm for MDSP, but we opted not to include it in our evaluations due to the difficulty of reproducing their results, as similarly highlighted in \cite{Xiong2024exact}.
The source code for EMOS \cite{jiang2023exact}, and the BIB-based algorithms \cite{Xiong2024exact}  is available to the public on GitHub\footnote{https://github.com/scikit-learn-contrib/radius\_clustering.git}.

All algorithms were implemented in C++ and compiled with gcc 7.1.0 using the `-O3' optimization flag. We employed HIGHS as the LP solver for our algorithms. The experiments were conducted on CentOS Linux 7.6.1810, utilizing an Intel(R) Xeon(R) w5-3435X  operating at 3.10 GHz with 256 GB of RAM. Our code, along with the experimental results, can be found here \footnote{https://github.com/yunanwubei/Exact-aigorithm.git}. Each trial was allocated a time limit of five hours, consistent with the settings in \cite{jiang2023exact} and \cite{Xiong2024exact}.

\subsection{Results and Analysis}

\begin{table}[!h]
\centering

\label{table1}
\resizebox{\columnwidth}{!}{
\begin{tabular}{>{\raggedright\arraybackslash}p{0.05\linewidth}>{\raggedleft\arraybackslash}p{0.07\linewidth}|cccccc}
\toprule
\multicolumn{8}{c}{\textbf{Dataset: T1}}\\
\midrule
$|V|$ & $|E|$ & ParDS& BIBLP& BIBCO& BIBLP-IF& BIBCO-IF& EMOS\\
\midrule
150 & 150& \textbf{0.01}& \textbf{0.01}& \textbf{0.01}& \textbf{0.01}& 623.14(2)& 12.38
\\
 150& 250& \textbf{0.01}& \textbf{0.01}& 1.44& 0.02& 26.2& 1206.75
\\
 150& 500& \textbf{1.74}& 3.86& 2320.74(5)& 7.39& 1734.28(3)& 3000(10)
\\
 150
& 750& 8.28& \textbf{5.76}& 2469.9(6)& 8.6& 1048.95& 3000(10)
\\
 150& 1000& 24.9& \textbf{18.06}& 2267.68(4)& 16.88& 690.7(1)& 3000(10)
\\
 150
& 2000& 13.31& 7.52& 63.31& \textbf{5.79}& 9.88& 3000(10)
\\
 150& 3000& 19.3& 6.87& 15.75& 5.71& \textbf{2.7}& 176.54
\\
 200& 250& \textbf{0.01}& \textbf{0.01}& 0.14& \textbf{0.01}& 2737.01(8)& 3000(10)
\\
 200& 500& \textbf{1.3}& 3.59& 2972.11(9)& 12.99& 3000(10)& 3000(10)
\\
 200& 750& \textbf{41.67}& 121.13& 3000(10)& 134.1& 3000(10)& 3000(10)
\\
 200& 1000& \textbf{428.22}& 440.8& 3000(10)& 681.81& 3000(10)& 3000(10)
\\
 200& 2000& 1674.39& 1085.96& 3000(10)& \textbf{1024.05}& 3000(10)& 3000(10)
\\
 200& 3000& 742.7& 500.7& 3000(10)& \textbf{391.53}& 3000(10)& 3000(10)
\\
 250& 250& \textbf{0.01}& \textbf{0.01}& \textbf{0.01}& \textbf{0.01}& 3000(10)& 3000(10)
\\
 250& 500& \textbf{2.28}& 9.64& 2759.28(9)& 12.55& 3000(10)& 3000(10)
\\
 250& 750& \textbf{200.66}& 678.11(1)& 3000(10)& 829.69(1)& 3000(10)& 3000(10)
\\
 250& 1000& \textbf{2052.42(5)}& 2383.03(7)& 3000(10)& 2378.91(7)& 3000(10)& 3000(10)
\\
 250& 2000& \textbf{2893.88(9)}& 3000(10)& 3000(10)& 3000(10)& 3000(10)&3000(10)
\\
 300& 300& \textbf{0.01}& \textbf{0.01}& \textbf{0.01}& \textbf{0.01}& 3000(10)& 3000(10)
\\
 300& 500& \textbf{1.46}& 4.05& 2857.65(9)& 3.14& 3000(10)& 3000(10)
\\
 300& 750& \textbf{211.74}& 1592.73(3)& 3000(10)& 2309.31(7)& 3000(10)& 3000(10)
\\
 300& 1000& \textbf{2854.51(9)}& 3000(10)& 3000(10)& 3000(10)& 3000(10)&3000(10)
\\
 800& 1000& \textbf{7.42}& 16.79& 3000(10)& 20.26& 3000(10)& 3000(10)
\\
 1000& 1000& \textbf{0.01}& \textbf{0.01}& \textbf{0.01}& \textbf{0.01}& 3000(10)&3000(10)
\\
\midrule
\multicolumn{8}{c}{\textbf{Dataset: Network Repository}}\\
\midrule
Class & num& ParDS& BIBLP& BIBCO& BIBLP-IF& BIBCO-IF& EMOS\\
\midrule
 bio& 28& \textbf{274.73}& 397.96(1)& 958.8(3)& 3600(12)& 6004.61(20)& 5867.16(19)\\
 eco& 6& \textbf{0.01}& \textbf{0.01}& \textbf{0.01}& \textbf{0.01}& \textbf{0.01}& \textbf{0.01}
\\
 econ& 11& 7.34& 3.92& 1200(4)& \textbf{1.44}& 1500.03(5)& 1202.83(4)\\
 email& 5& \textbf{9.38}& 14.73& 300(1)& 300(1)& 900(3)& 1500(5)\\
 p2p& 4& \textbf{0.01}& 0.08& 6.18& 0.02& 1200(4)& 722.42(2)\\
 ia& 16& \textbf{11.27}& 13.94& 313.81(1)& 501.64(1)& 1800.04(6)& 1200.01(4)\\
 prox& 6& 0.04& \textbf{0.01}& 0.49& \textbf{0.01}& 0.22& 18000(6)\\
 soc& 11& \textbf{181.33}& 900.05(3)& 1205.92(4)& 901.95(3)& 2128.31(7)& 2700(9)\\
 infect& 2& \textbf{0.01}& \textbf{0.01}& 0.03& 201.64& 300(1)& 300(1)\\
 Erdos& 7& \textbf{0.01}& \textbf{0.01}& 0.73& \textbf{0.01}& 42.21& \textbf{0.01}
\\
 others& 6& \textbf{300.02(1)}& 361.49(1)& 347.71(1)& 900.01(3)& 1500(5)& 900.02(3)\\
 $G_{800}$& 4 & \textbf{0.34}& 1.33& 1200(4)& 1.8& 1200(4)&1200(4)\\
 $G_{2000}$& 4 & \textbf{490.22}& 1200(4)& 1200(4)& 1200(4)& 1200(4)&1200(4)\\
 $G_{3000}$& 3 & \textbf{4.73}& 6.63& 900(3)& 6.34& 900(3)&900(3)\\
\midrule

 \multicolumn{8}{c}{\textbf{Dataset: BD3/BD6}}\\
\midrule
Class & num & ParDS& BIBLP& BIBCO& BIBLP-IF& BIBCO-IF& EMOS\\
\midrule
 BD3 & 18& 2586.34(4)& 1456.34(3)& 1097.24(3)& 2901.8(5)& \textbf{980.51(3)}&5400(18)\\
 BD6 & 18& \textbf{0.08}& 0.17& 1.38& 0.19& 0.67&5700(18)\\
 \midrule
 
\multicolumn{8}{c}{\textbf{Dataset: UDG}}\\
\midrule
$|V|$ & Type & ParDS& BIBLP& BIBCO& BIBLP-IF& BIBCO-IF& EMOS\\
\midrule
100 & A & \textbf{0.01}& \textbf{0.01}& \textbf{0.01}& \textbf{0.01}& 0.22& \textbf{0.01}
\\
100 & B & \textbf{0.01}& \textbf{0.01}& \textbf{0.01}& \textbf{0.01}& \textbf{0.01}& \textbf{0.01}
\\
250 & A & \textbf{0.01}& \textbf{0.01}& \textbf{0.01}& 0.06& 1556.3(4)& 198.24
\\
250 & B & \textbf{0.01}& \textbf{0.01}& \textbf{0.01}& 0.02& 62.62& 2.59
\\
500& A& 0.18& \textbf{0.02}& 1.39& 0.68& 3000(10)& 3000(10)
\\
 500& B& 0.16& \textbf{0.01}& 0.1& 0.05& 1271.82(3)& 1674.09(5)\\
 800& A& 1.82& \textbf{0.03}& 0.91& 0.38& 3000(10)& 3000(10)
\\
 800& B& 2.99& \textbf{0.16}& 104.58& 5.39& 3000(10)& 3000(10)
\\
 1000& A& 4.17& \textbf{0.06}& 3.34& 0.39& 3000(10)& 3000(10)
\\
 1000& B& 16.33& \textbf{0.42}& 697.86& 30.31& 3000(10)& 3000(10)
\\
 
\midrule
\end{tabular}}
\caption{Running time comparison of six algorithms on four datasets. The running time is measured in minutes (min).}\label{table1}
\end{table}

Table \ref{table1} presents the running times of all evaluated algorithms across four datasets, with each instance constrained by a 5-hour time limit, where the notation`($k$)' signifies that the algorithm was unable to find a solution within five hours for $k$ instances in that category. The fastest recorded time in each category is highlighted in bold. Considering the potential influence of the hardware environment on running speeds, we have standardized all recorded running times of less than 0.01 minutes to 0.01 minutes. As noted in \cite{Xiong2024exact}, we excluded T1 graphs with fewer than 100 vertices and UDG graphs with fewer than 50 vertices. Furthermore, for the other omitted T1 categories, no algorithm was able to successfully solve any of the instances.




Overall, ParDS showcased exceptional performance on the datasets,  
achieving the fastest runtime in 36 categories and outperforming all competitors in 22 without any ties.  In contrast, the leading alternative among the competing algorithms, BIBLP, achieved the fastest runtime in 22 categories, with strictly superior performance in only eight instances.

On the T1 dataset, ParDS achieved the fastest runtime in 19 out of 25 categories, significantly surpassing its competitors. In comparison, BIBLP, the second-best performer, secured the top speed in only eight categories. Notably, all other algorithms completely failed in the V250E750 and V300E750 categories, while ParDS excelled by successfully solving one instance in each of these two challenging categories. Moreover, ParDS fully resolved all instances in the V250E750 class, marking a significant advancement that no prior algorithm has achieved for this category of graphs.

For power-law graphs in Network Repository, which are commonly found in real-world networks and pose significant challenges for approximation \cite{DBLP:journals/tcs/GastHK15}, ParDS once again exhibited exceptional performance. It recorded the fastest runtime in 12 out of 14 categories, while its closest competitor, BIBLP, excelled in only four categories. Notably, in the `USGrid' instance from the `others' category, ParDS outperformed competing algorithms by an astonishing factor of over 3,411 times.


In the BD6 category, ParDS demonstrated outstanding performance, outperforming other algorithms in 16 out of 18 instances. In contrast, within the BD3 dense-graph category, the performance ranking is as follows: BIBCO-IF $>$ BIBCO $>$ BIBLP $>$ ParDS $>$ BIBLP-IF. This ranking correlates with LP invocation counts, ordered as: BIBCO $<$ BIBLP $<$ ParDS. The LPs associated with dense graphs involve large constraint matrices, resulting in significantly longer solution times. Consequently, methods that reduce LP operations may offer a competitive advantage over ParDS in the realm of dense graphs.

In UDG, the performance ranking is as follows: \( \text{BIBLP} > \text{ParDS} > \text{BIBCO} > \text{BIBLP-IF} > \text{BIBCO-IF} \). In this dataset, while BIBLP outperforms ParDS, both algorithms demonstrate relatively short overall runtimes. Our algorithm employs more sophisticated strategies, which introduce additional computational overhead; however, the benefits gained from these strategies are limited in this context. This limitation is likely due to the nature of the graphs, which are relatively easy to solve and provide minimal opportunities for further optimization. Consequently, the increased complexity of our approach results in longer runtimes without significant improvements in performance.

\begin{table}[h]
\centering

\label{table2}
\footnotesize
\resizebox{\columnwidth}{!}{\begin{tabular}{c|c|ccc|ccc}
\midrule
      \multicolumn{2}{c}{\textbf{Dataset: T1}} &\multicolumn{3}{c}{Time} &\multicolumn{3}{c}{s\_num}\\
\midrule
Graph  & OPT&ParDS& ParDS-LP& BIBLP &ParDS& ParDS-LP& BIBLP \\
\midrule

 V300E750\_0
 & 61&4474*& \textbf{5086}& 18000
&180646*& \textbf{184889}&failed\\
  V300E750\_1
 & 59
&243*& \textbf{350}& 5814
&8130*& \textbf{8225}&346749
\\
   V300E750\_2
 & 60&940*& \textbf{1263}& 18000
&38065*& \textbf{38335}&failed\\
    V300E750\_3
 & 57
&61*& \textbf{94}& 352
&3016*& \textbf{3023}&21160
\\
     V300E750\_4
 & 60
&1181*& \textbf{1569}& 13217
&42977*& \textbf{43162}&929033
\\
      V300E750\_5
 & 59
&1676*& \textbf{1919}& 11830
&61870*& \textbf{61938}&739481
\\
       V300E750\_6
 & 60
&968*& \textbf{1325}& 8942
&43162*& \textbf{43437}&586353
\\ 
       V300E750\_7
 & 60
&57*& \textbf{83}& 618
&2481*& \textbf{2584}&42810
\\
        V300E750\_8
 & 60&3000*& \textbf{3370}& 18000
&111471*& \textbf{112131}&failed\\
         V300E750\_9
 & 60
&99*& \textbf{140}& 788
&4212*& \textbf{4246}&58125
\\
          V250E750\_0
 & 45&3024*& \textbf{3471}& 18000
&152478*& \textbf{162559}&failed\\
           V250E750\_1
 & 44
&294*& \textbf{482}& 3804
&12356*& \textbf{13815}&508001
\\
           V250E750\_2
 & 43
&372*& 601& \textbf{409}
&18895*& \textbf{20789}&47351
\\
           V250E750\_3
 & 44
&553*& \textbf{860}& 2574
&25008*& \textbf{26717}&265277
\\
           V250E750\_4
 & 43
&904*& \textbf{1089}& 1769
&44531*& \textbf{52036}&193535
\\
           V250E750\_5
 & 44
&3908*& 4111& \textbf{3698}
&172105*& \textbf{204993}&386787
\\
           V250E750\_6
 & 43
&270*& 488& \textbf{246}
&13412*& \textbf{15025}&32436
\\
           V250E750\_7
 & 45
&275*& \textbf{348}& 3392
&16642*& \textbf{19369}&515650
\\
           V250E750\_8
 & 45
&2247*& \textbf{2718}& 5880
&115072*& \textbf{120028}&715123
\\
           V250E750\_9
 & 44
&188*& \textbf{334}& 910
&10290*& \textbf{12314}&115700
\\

\midrule
           
\end{tabular}}
\caption{Experimental results: Performance comparison of ParDS-LP, BIBLP, and ParDS} \label{table2}
\end{table}

\subsection{Evaluating the Reduction Rules and Lower Bound}
As previously discussed, the ParDS and BIBLP algorithms are the first two top performers on the benchmark dataset. Both use a hybrid framework that integrates reduction rules with an LP approach. While BIBLP employs a conventional LP solver, ParDS utilizes an enhanced LP technique. To assess the impact of our reduction rules and enhanced LP, we developed a variant of ParDS called ParDS-LP, which replaces the enhanced LP component with the conventional LP method used in BIBLP.

We conducted a series of experiments to evaluate and compare the performance of ParDS-LP against both BIBLP and ParDS. The comparison with BIBLP focuses on the efficacy of our reduction rules, while the comparison with ParDS assesses the contribution of our enhanced LP technique. Performance was measured based on running time and the number of search nodes, quantified as vertex selection counts. To ensure meaningful results, we selected 20 graph instances that featured longer solution times and denser sampling, thereby emphasizing each algorithm's ability to reduce the search space.

Table \ref{table2} summarizes the results, where `OPT', `time', and `s\_num' denote the solution size, the running time, and vertex selection counts, respectively. For each pairwise comparison, the better result between ParDS-LP and BIBLP is shown in bold, while the better result between ParDS and ParDS-LP is marked with an asterisk. The data indicate that ParDS-LP consistently outperforms BIBLP, achieving faster runtimes in 17 out of 20 instances and lower vertex selection counts in all instances. In contrast, BIBLP achieved the best runtime in only two cases and never produced the lowest vertex selection count. These results clearly demonstrate the effectiveness of our reduction rules. In the comparison between ParDS and ParDS-LP, ParDS shows superior performance across all instances, both in terms of runtime and vertex selection counts. This validates the role of our enhanced LP in improving the overall efficiency of ParDS.

\section{Conclusion} \label{sec-7}

We introduce ParDS, an exact branch-and-bound algorithm designed for MDS. This algorithm features two key innovations: an advanced LP approach that provides a tighter lower bound and a set of novel reduction rules for simplifying instances. ParDS has successfully resolved 16 previously unsolved instances within five hours and achieved the fastest runtime across 70\% of graph categories. Future research will focus on expanding the LP framework to support weighted and dynamic network applications.



\bibliography{ref}

\end{document}